\documentclass[reprint,twocolumn,superscriptaddress,showpacs,nofootinbib,notitlepage,preprintnumbers,aps,pra,floatfix]{revtex4-2}
\usepackage[utf8]{inputenc}
\usepackage{graphicx}
\usepackage{latexsym,amsmath,amssymb,lmodern,float,url}
\usepackage{qcircuit}
\usepackage{natbib}
\usepackage{color}
\usepackage{microtype}
\usepackage{bbold}
\usepackage{slashed}
\usepackage{multirow}
\usepackage{tikz}
\usepackage{xcolor}
\usetikzlibrary{shapes}
\usepackage{adjustbox}
\usepackage{braket}
\usepackage[linesnumbered,ruled,vlined]{algorithm2e}

 \usepackage[normalem]{ulem}
\usepackage{amsthm}
\usepackage{pgfmath}

\newcommand{\sref}[1]{Sec.~\ref{#1}}

\usepackage[colorlinks=true,backref=false, linktocpage=true,
citecolor=blue,urlcolor=blue,linkcolor=blue,pdfpagemode=UseOutlines]{hyperref}
\hypersetup{%
 bookmarksnumbered=true,
 pdftitle = {},
 pdfsubject = {},
 pdfauthor = {},
 pdfkeywords = {}
}
\usepackage{cleveref}
\Crefname{appendix}{App.}{Apps.}
\Crefname{equation}{Eq.}{Eqs.}
\Crefname{figure}{Fig.}{Figs.}

\newcommand{\yint}{3.20(13)}
\newcommand{\slope}{10.77(3)}
\newcommand{\yintm}{2.193(11)}
\newcommand{\slopem}{8.621(7)}
\pgfmathsetmacro{\logthreeoften}{ln(10)/ln(3)}
\pgfmathsetmacro{\fullcomplexity}{10 - (3.5 * 8.621) / \logthreeoften }
\pgfmathsetmacro{\householdercomplexity}{6 - (2.5 * 0.5 * 10.77) / \logthreeoften }
\pgfmathsetmacro{\slopethree}{10.77 / \logthreeoften }
\pgfmathsetmacro{\slopemthree}{8.621 / \logthreeoften }
\def\bZ {\mathbb{Z}}

\newtheorem{lemma}{Lemma}

\usepackage{silence}
\begin{document}
 \preprint{FERMILAB-PUB-25-0002-SQMS-T}
 \title{Synthesis of Single Qutrit Circuits from Clifford+R}
 \author{Erik J. Gustafson}
\email{egustafson@usra.edu}
\affiliation{Superconducting and Quantum Materials System Center (SQMS), Batavia, Illinois, 60510, USA.}
\affiliation{Fermi National Accelerator Laboratory, Batavia,  Illinois, 60510, USA}
\affiliation{Quantum Artificial Intelligence Laboratory (QuAIL),
NASA Ames Research Center, Moffett Field, CA, 94035, USA}
\affiliation{USRA Research Institute for Advanced Computer Science (RIACS), Mountain View, CA, 94043, USA}
\author{Henry Lamm}
\email{hlamm@fnal.gov}
\affiliation{Superconducting and Quantum Materials System Center (SQMS), Batavia, Illinois, 60510, USA.}
\affiliation{Fermi National Accelerator Laboratory, Batavia, Illinois, 60510, USA}
\author{Diyi Liu}
\email{liu00994@umn.edu}
\affiliation{School of Mathematics, University of Minnesota-Twin Cities, Minneapolis, MN, 55414, USA}
\author{Edison M. Murairi}
\email{emurairi@fnal.gov}
\affiliation{Superconducting and Quantum Materials System Center (SQMS), Batavia, Illinois, 60510, USA.}
\affiliation{Fermi National Accelerator Laboratory, Batavia, Illinois, 60510, USA}
\affiliation{Department of Physics, The George Washington University, Washington, DC  20052, USA}
\author{Shuchen Zhu}
\email{shuchen.zhu@duke.edu} 
\affiliation{Department of Mathematics, Duke University, Durham, NC 27708, USA}
\date{\today}

\begin{abstract}
We present two deterministic algorithms to approximate single-qutrit gates. These algorithms utilize the Clifford + $\mathbf{R}$ group to find the best approximation of diagonal rotations. The first algorithm exhaustively searches over the group; while the second algorithm searches only for  Householder reflections. 
The exhaustive search algorithm yields an average $\mathbf{R}$ count of $\yintm + \slopem \log_{10}(1 / \varepsilon)$, albeit with a time complexity of $\mathcal{O}(\varepsilon^{\pgfmathprintnumber[fixed,precision=2]{\fullcomplexity}})$. The Householder search algorithm results in a larger average $\mathbf{R}$ count of $\yint + \slope \log_{10}(1 / \varepsilon)$ at a reduced time complexity of $\mathcal{O}(\varepsilon^{\pgfmathprintnumber[fixed,precision=2]{\householdercomplexity}})$, greatly extending the reach in $\varepsilon$. These costs correspond asymptotically to 35\% and 69\% more non-Clifford gates compared to synthesizing the same unitary with two qubits. Such initial results are encouraging for using the $\mathbf{R}$ gate as the non-transversal gate for qutrit-based computation.
\end{abstract}
\maketitle
\section{Introduction}\label{sec:introduction}
Various fields anticipate quantum computing will tackle problems that are intractable for classical computers, but the large-scale architecture is still an active area of study. While most devices are qubit-based, many have access to higher levels and thus could be run as qudit-based platforms including: trapped ions~\cite{Ringbauer:2021lhi,Low:2023dlg,Nikolaeva:2024wxl}, transmons~\cite{Galda:2021ega,Goss:2022bqd,Luo:2022pxs,PRXQuantum.4.030327, Goss:2023frd,Cao:2023ekj,Seifert:2023ous,Iiyama:2024uos,Wang:2024xbz}, Rydberg arrays~\cite{Kruckenhauser:2022qfi,Cohen:2021axm}, photonic circuits~\cite{Chi:2022uql}, cold atoms~\cite{Kasper:2020mun,Ammenwerth:2024glw}, and superconducting radio frequency (SRF) cavities~\cite{Roy:2024uro}. While experimentally more challenging, there are advantages to developing qudit-based systems from an algorithmic perspective due to their enhanced effective connectivity, as native single-qudit $SU(d)$ rotations replace non-local multi-qubit circuits~\cite{jankovic2024noisy,Nikolaeva:2021rhq,Mansky:2022bai,Murairi:2024xpc}.
In practice, this allows for lower gate fidelities for the same algorithmic fidelity~\cite{wang2020qudits,iqbal2024qutrit,Majumdar:2024sms,Nikolaeva:2022wmq,Iiyama:2024uos,Saha:2023epn,Champion:2024ufp,kiktenko2020scalable, roy2023two,PRXQuantum.2.020311, zhu2024unified}. Such potential has lead to application-specific research in qudits across fields such as: material science~\cite{Malpathak:2024weo, kaxiras2019quantum,lin2019mathematical,Babbush:2023aww,chicco2023proof}
numerical optimization  \cite{Tancara:2024vck,Bottrill:2023lyt,Bravyi:2020pur,Saha:2022tol}, 
condensed matter~\cite{PhysRevA.110.062602,ogunkoya2024qutrit,choi2017dynamical,BassmanOftelie:2022hfz,Young:2023zxu}, and particle physics~\cite{Gustafson:2021jtq,Gustafson:2021qbt,Gonzalez-Cuadra:2022hxt,Popov:2023xft,nguyen2023simulating,Calajo:2024qrc,Illa:2024kmf, Zache:2023cfj,Gustafson:2024kym,Spagnoli:2025etu}.

 Regardless of the qudit dimensionality, reaching the goal of quantum utility requires fault-tolerant gate synthesis. That is, one identifies a finite set of generators that can efficiently approximate unitary operations to any required precision and support quantum error correction for the logical gates in the set. Fundamentally though, Eastin-Knill theorem~\cite{Eastin_2009} prevents a universal gate set that is also transversal, i.e. not all logical gates can be implemented in parallel across the physical qudits. This constrains the gate sets $\mathcal{G}$ for large-scale computation. Further, the nontransversal gate counts $N_{\mathcal{G}}$ dominates the computation costs \cite{nielsen_chuang_2010} but recent evidence suggests they may not be as expensive as thought~\cite{Gidney:2024alh,wills2024constant}. Here, we take the first steps beyond qubits by considering the fault-tolerant gate synthesis of $d=3$ qudits, called \emph{qutrits}.

The prevalence of the qubits extends to fault-tolerant gate synthesis. 
There, the Clifford group extended by $\mathbf{T}=\text{Diag}(1,e^{i\pi/4})$ -- denoted $(\mathbf{C}+\mathbf{T})_2$ -- is the leading choice. However, novel sets with non-Clifford transversal operations \cite{yoder2017universal,Zhu:2023xfg,rengaswamy2020optimality} or those based on groups larger than the Clifford group also exist~\cite{Parzanchevski_2018,blackman2022fastnavigationicosahedralgolden,Kubischta:2023nlb}. Once this finite gate set is selected, one must map all other circuit primitives to a gate set \emph{word}. While $(\mathbf{C}+\mathbf{T})_2$ does not result in the shortest word lengths compared to larger groups~\cite{PhysRevA.88.012313,Low:2018dby, blackman2022fastnavigationicosahedralgolden,Kubischta:2023nlb}, it has well-established error correction schemes and experimental demonstrations in contrast to other gate sets and codes. 
In the case of qutrits, three options have dominated the literature, which extend the qutrit-Cliffords $\mathbf{C}_3$, although others exist~\cite{Evra_2022,Kubischta:2024fuf}.  The first uses with the generalized $\mathbf{T}_3=\operatorname{Diag}(1, \omega, \omega^2)$~\cite{Howard:2012wif,Campbell:2012olh,Prakash:2018rrp,Prakash:2021axb,Amaro-Alcala:2024sfk} where $\omega=e^{2\pi i / 3}$ is the third root of unity. The other two common extensions are $\mathbf{D}(a,b,c) = \mathrm{D
iag}(\pm\xi^{a}, \pm\xi^{b}, \pm\xi^{c})$~\cite{Kalra:2023llu,Glaudell:2024tvz,Evra:2024bkr} where $\xi=e^{2\pi i / 9}$, and the metaplectic $\mathbf{R}=\operatorname{Diag}(1, 1, -1)$~\cite{Bocharov:2018zfk,Glaudell:2022vvu} gates. While the metaplectic set $(\mathbf{C}+\mathbf{R})_3$ is strictly a subset of $(\mathbf{C}+\mathbf{T})_3$~\cite{Glaudell:2022vvu}, it may prove more practical in hardware and thus warrants study. Thus, in this work we will study the synthesis of $SU(3)$ unitaries using $(\mathbf{C}+\mathbf{R})_3$.

This work is organized as follows. We briefly review qutrit-based computation and ring-based gate synthesis in  Sec.~\ref{sec:general-qutrit-background}.  This leads into Sec.~\ref{sec:algorithm-one} and~\ref{sec:algorithm-two} where two algorithms are constructed for approximating diagonal single-qutrit gates: the exhaustive search and the Householder search. This is followed by a method for determining the $(\mathbf{C}+\mathbf{R})_3$ word for the approximation in Sec.~\ref{sec:exact-synthesis}. Numerical studies are found in Sec.~\ref{sec:numerical-results} before concluding in Sec.~\ref{sec:conclusion}.

\section{Theoretical Background}\label{sec:general-qutrit-background}
Qutrit systems have a basis of three-level states $|0\rangle,~|1\rangle,$ and $|2\rangle$; their Hilbert spaces scale as $3^N$ which is a polynomial increase compared to the qubits' $2^N$. A universal—though not fault-tolerant—gate set for qutrits can be built from single-qutrit rotations and an entangling gate. One such set is the 18 two-level Givens rotations,
\begin{equation} R^{\alpha}_{(b, c)}(\theta) = e^{-i \theta / 2 \sigma^{\alpha}_{(b,c)}}, \end{equation}
where $\sigma^{\alpha}_{(b,c)}$ is the Pauli matrix $\sigma^{\alpha}=\{X,Y,Z\}$ acting on the $|b\rangle-|c\rangle$ subspace, combined with the $\textsc{CSum}$ gate:
\begin{equation} \textsc{CSum}|i\rangle|j\rangle = |i\rangle|i \oplus_3 j\rangle. \end{equation}
While $\textsc{CSum}$ is part of qutrit Clifford group, $R^{\alpha}_{(b, c)}(\theta)$ are not. 
Therefore if one wants to implement the $R^{\alpha}_{(b, c)}(\theta)$ rotations using a fault-tolerant gate set, one has to approximately synthesize these rotations using a finite set of gates such as $(\mathbf{C} + \mathbf{R})_3$. 

The Solovay-Kitaev theorem~\cite{Dawson:2005blj} states that any $d$-dimensional qudit gate $U \in SU(d)$ that cannot be exactly synthesized from a universal fault-tolerant gate set $\mathcal{G}$ can still be approximated by a gate $V\in\mathcal{G}$, with $||U - V|| \leq \varepsilon$ and $N_\mathcal{G}$ scaling as $\mathcal{O}(\log_d^a(1/\varepsilon))$.

Optimal algorithms correspond to the case where $a = 1$~\cite{Harrow2002,Bourgain2012}. Further, by considering the geometric structure of hyperspheres covering $SU(d)$, one can show that there always exists a unitary $U \in SU(d)$ whose $\varepsilon$-approximation $V$ requires at least~\cite{GLAUDELL201954,Prakash:2021axb}:
\begin{equation}
 \label{eq:ftgate_bound}
    N_\mathcal{G} \gtrsim \frac{\ln\left(A\right) + (d^{2}-1)\ln\left(\frac{1}{\varepsilon}\right)}{\ln(d(d-1))},
\end{equation}
where 
\begin{equation}
A = \frac{\sqrt{2^{d-1}d}[d(d-1)-1]\,\Gamma(\frac{d^2-1}{2})}{d^4\pi^{3/2(d-1)}G(d+1)},
\end{equation}
and $G(n) = \prod_{k=1}^{d-1} k!$ is the Barnes G-function.
For the case of qubits ($d = 2$), Eq.~(\ref{eq:ftgate_bound}) predicts:
\begin{align}
    N_\mathcal{G} &= 3\log_2(1/\varepsilon) - 5.65 \notag \\
        &= 9.97\log_{10}(1/\varepsilon) - 5.65.
\end{align}
For qutrits ($d = 3$), we find:
\begin{align}
    N_\mathcal{G} &= 4.9\log_3(1/\varepsilon) - 2.16 \notag \\
        &= 10.27\log_{10}(1/\varepsilon) - 2.16.
\end{align}

The described properties are independent of a specific $\mathcal{G}$, meaning that specific instances may exhibit different performance. We focus here on $\mathcal{G}=(\mathbf{C}+\mathbf{R})_3$ which is generated by the qutrit Hadamard $H$, phase gate $S$, and $\mathbf{R}$, following the notation of Ref.~\cite{Kalra:2023llu}:
\begin{align}
    \label{eq:clifford-r-presentation}
    &\phantom{xxxxx}H=\frac{1}{i\sqrt{3}}\begin{pmatrix}
        1 & 1 & 1\\
        1 & \omega & \omega^2\\
        1 & \omega^2 & \omega\\
    \end{pmatrix},\notag\\
    &S=\operatorname{Diag}(1,\omega,1),~\mathbf{R}=\operatorname{Diag}(1, 1, -1).
\end{align}
Other gates in $\mathbf{C}_3$  that will prove useful include: 
\begin{align}
   X&= \begin{pmatrix}
        0 & 0 & 1\\
        1 & 0 & 0\\
        0 & 1 & 0\\
    \end{pmatrix},~D(a,b,c) = \mathrm{diag}(\omega^{a}, \omega^{b}, \omega^{c}),
\end{align}
where $a,b,c\in \{0, 1, 2\}$. The decomposition of $D(a,b,c)$ into $H,S$ can be aided by the relations: 
\begin{align}
   ~H^{\dagger}=H^3,~X_{(0,1)}=HS^2H^2SH^{\dagger},\notag\\
    X_{(1,2)}=H^2,~\text{and}~X = H^{\dagger} S H^2 S^2 H^{\dagger},
\end{align}
As an example
\begin{equation}
D(1,2,1)=X_{(0,1)}SXSX_{(1,2)}S^2.
\end{equation}

Determining efficient synthesis is an ongoing area of research, even for qubits.  The best algorithms rely upon insight from number theory. One can show that any $V\in\mathcal{G}$ consists of matrix entries in a ring $\mathcal{R}$. This was proven and then used to perform exact single-qubit synthesis in $(\mathbf{C}+\mathbf{T})_2$~\cite{10.5555/2535649.2535653}.  The extension to approximate synthesis requires determining a $V$ that is within distance $\varepsilon$ of the desired gate.

Identifying these approximations requires solving a Diophantine equation. In general this is NP-complete~\cite{manders1978np} and thus finding the shortest word is often difficult. Luckily, this need not prevent subclasses of gates from being approximated efficiently. In particular, a probabilistic number-theoretic method to approximate diagonal single-qubit gates was first introduced in~\cite{Ross:2014okw}.  Since any single-qubit gate can be exactly represented by 3 diagonal gates and $\mathbf{C}_2$, it was demonstrated that $U\in SU(2)$ could be approximated with $N_\mathcal{G}=3\log_p (1/\varepsilon^3)$ with a gate set associated to a prime $p$.  For $(\mathbf{C}+\mathbf{T})_2$ one finds $p=2$ while other ``golden gate sets" can be used to reduce this bound up to $\frac{7}{3}\log_{59}(1/\varepsilon^3)$~\cite{blackman2022fastnavigationicosahedralgolden}. Further improvements have been made for $(\mathbf{C}+\mathbf{T})_2$, with the state-of-the-art being the repeat-until-success method of~\cite{PhysRevLett.114.080502}. Ultimately, finding ways to decompose and approximate arbitrary gates more efficiently than by diagonal matrices remains an open problem. 
   
Similarly, $(\mathbf{C}+\mathbf{R})_3$ can be related to a ring. Starting from the ring of Eisenstein integers $\mathcal{R}_{3}=\{a_{0}+a_{1}\omega~|~a_{i}\in \bZ,~\omega=e^\frac{{2\pi i}}{3}\}$ one localizes\footnote{Localization extends a ring while introducing a division-like operation.} it to obtain the ring $\mathcal{R}_{3,\chi}  =\{\frac{a}{\chi^{f}}~|~a\in \mathcal{R}_{3},~f\in \mathbb{N}_0 \}$ where $\chi = 1 + 2\omega = \sqrt{-3}$.  
Inspecting the generators (\Cref{eq:clifford-r-presentation}), one sees that all their entries are in $\mathcal{R}_{3,\chi}$. 
Therefore, the set generated by the $(\mathbf{C}+\mathbf{R})_3$ is the unitary group over the $\mathcal{R}_{3,\chi}$ ring, denoted by $U\left(3, \mathcal{R}_{3,\chi}\right)$~\cite{Kalra:2023llu,PhysRevA.93.012313}. Thus, any matrix $V\in(\mathbf{C}+\mathbf{R})_3$ has the form
\begin{align}
    V = \frac{1}{\chi^f}\begin{pmatrix}
        x_1 & y_1 & z_1\\
        x_2 & y_2 & z_2\\
        x_3 & y_3 & z_3
    \end{pmatrix},
\end{align}
where $x_i,y_i,z_i \in \mathcal{R}_3$, and $f\in\mathbb{N}_0$. 
An important consequence of this observation is that synthesis over this gate set can be performed with optimal $\mathcal{O}(\log_3 \frac{1}{\varepsilon})$. 
The remaining challenge lies in developing a constructive algorithm that reaches this complexity while minimizing constant factors.   
Such an algorithm can be decomposed into two steps. First, find a sufficient approximation $V\in U\left(3, \mathcal{R}_{3,\chi}\right)$.
Then, determine the $(\mathbf{C}+\mathbf{R})_3$ word that exactly synthesizes $V$. 

Any single-qutrit unitary is a product of a global phase and an $SU(3)$ matrix. Noting that a global phase may be synthesized by a  modification of our first algorithm, we focus only on $SU(3)$ matrices which by leveraging $\mathbf{C}_3$ rotations can be synthesized from
\begin{align}
    R^Z_{(0,1)}(\theta)=\operatorname{Diag}(e^{-i\theta/2},e^{i\theta / 2}, 1).
\end{align}
Thus goal is to find $V\in U(3,\mathcal{R}_{3,\chi})$ such that 
\begin{equation}
\label{eq:rznorm}
    || R^Z_{(0,1)}(\theta) - V || \leq \varepsilon,
\end{equation}
given a choice of $\theta$ and error, $\varepsilon > 0$. 
We will use the Frobenius norm\footnote{The Frobenius norm is $\|A\|^2_{F} = \text{tr}\left(A^\dagger A\right) = \sum_{i,j} |A_{ij}|^2$.} throughout this work. Furthermore, anticipating the large cost of implementing non-Clifford gates, we use the number of $\mathbf{R}$ gates, $N_\mathbf{R}$ as the quantum complexity of the synthesis problem. 

To solve \Cref{eq:rznorm}, we will present two algorithms. The first algorithm conducts an exhaustive search over $(\mathbf{C}+\mathbf{R})_3$ group, yielding good results but incurring significant classical runtime. The second algorithm restricts its search to Householder reflection gates, improving classical complexity while requiring an increased $N_\mathbf{R}$. The following sections provide detailed explanations of both algorithms.

\section{Exhaustive search algorithm}\label{sec:algorithm-one}
To find a gate $V$ that satisfies \Cref{eq:rznorm}, it behooves us to expand the Frobenius norm:
\begin{align}
\label{eq:rzapprox}
    ||R_{(0,1)}^Z(\theta)-V||^2 &= \sum\limits_{j}\left(\sum\limits_{i} \left|R_{(0,1)}^Z(\theta)_{ij} - V_{ij}\right|^2\right)\notag\\&= \sum\limits_{j}||R_{(0,1)}^Z(\theta)_{j} - V_{j}||^2.
\end{align}
From this, we see that the norm can be decomposed into a sum over column vectors, which each contributes $\|R_{(0,1)}^Z(\theta)_{i} - V_{i}\|^2$. As a result, approximating $R_{(0,1)}^Z(\theta)$ may be reduced to approximating each of its column vectors with errors $\varepsilon_i$ ($i = 1,2,3$) such that $\sum_i \varepsilon^2_i \leq \varepsilon^2$.    

Consider a target unit vector $\mathbf{t}(j)=e^{i\alpha} \delta_{ij}$ with a single nonzero entry and $\alpha$ being a real number. 
To approximate it using a unit vector $\mathbf{v}$ imposes a condition on only one entry of $\mathbf{v}$:
\begin{align}
    \|\mathbf{t}(j) - \mathbf{v}\|^2 &= \sum\limits_{i} |e^{i\alpha}\delta_{ij} - v_i|^2\notag\\
    &= 2 - 2\,\operatorname{Re}\left(v_j\,e^{-i\alpha}\right).
\end{align}
Requiring this be bound by $\varepsilon_j^2$, we can apply it to \Cref{eq:rzapprox} to give:
\begin{align}
    \label{eq:conditions-diag}
    {\operatorname{Re}}\left(\frac{x_1}{\chi^f}e^{i\frac{\theta}{2}}\right) \geq &\eta(\varepsilon_1), \quad {\operatorname{Re}}\left(\frac{y_2}{\chi^f}e^{-i\frac{\theta}{2}}\right) \geq \eta(\varepsilon_2),\notag\\ 
    &{\operatorname{Re}}\left(\frac{z_3}{\chi^f}\right) \geq \eta(\varepsilon_3),
\end{align}
where $\eta(\varepsilon) := 1 - \varepsilon^2/2$. To derive an algorithm, it is useful to further simplify Eq.~\eqref{eq:conditions-diag}.  
To do so, we note
\begin{equation}
    \chi^{-f} =(-3)^{-f/2} = (-1)^{\lceil f/2 \rceil}3^{-\lceil f/2 \rceil} \chi^{\bar{f}},
\end{equation}
where $\bar{f} := f \mod 2$. 
By introducing the change of variables   
\begin{equation}
\label{eq:cov}
    x^\prime_1 = (-1)^{\lceil f/2\rceil} \chi^{\bar{f}} x_1
\end{equation} and applying similar transformations for $y_2$ and $z_3$, Eq. \eqref{eq:conditions-diag} can be rewritten as
\begin{align}
    \label{eq:approx-constraints-full-search}
    \operatorname{Re}(x^\prime_1 e^{\mathrm{i}\theta/2}) \geq 3^{\lceil f/2\rceil } \eta(\varepsilon_1),\notag\\
     \operatorname{Re}(y^\prime_2 e^{-\mathrm{i}\theta/2}) \geq 3^{\lceil f/2\rceil } \eta(\varepsilon_2),\notag\\
     \operatorname{Re}(z^\prime_3 ) \geq 3^{\lceil f/2\rceil } \eta(\varepsilon_3),
\end{align}
Moreover, the unitarity of $V$ imposes the bounds 
\begin{equation}
\label{eq:exunbounds}
    |x^\prime_1|^2, |y^\prime_2|^2, |z^\prime_3|^2  \leq 3^{2\lceil f/2\rceil}.
\end{equation}

The constraints of \Cref{eq:approx-constraints-full-search,eq:exunbounds} have straightforward geometric interpretations. Any complex number $z$ (e.g. $e^{i\theta/2}$) can be mapped to a vector in $\mathbb{R}^2$ as $z \mapsto (\operatorname{Re}(z), \operatorname{Im}(z))^T$. Similarly, any Eisenstein integer $x_1 + x_2\omega$ can be mapped to vectors $\mathbf{y} = (x_1 - \frac{x_2}{2}, \frac{x_2 \sqrt{3}}{2})^T$ in $\mathbb{R}^2$. As a result, the Eisenstein integers form an integer lattice $\mathcal{L}_1$ given by
\begin{equation}
    \label{eq:eis-lattice}
    \mathcal{L}_1 = \left\{ \mathbf{y} = B_1\,  \mathbf{x} \, \Big|\, \mathbf{x} \in \mathbb{Z}^2 \right\},
\end{equation}
where 
\begin{equation}
    \label{eq:eis-lattice-basis}
    B_1 = \begin{pmatrix}
        1 & -\frac{1}{2}\\
        0 & \frac{\sqrt{3}}{2}
    \end{pmatrix}.
\end{equation}

From Equations \eqref{eq:eis-lattice} and \eqref{eq:eis-lattice-basis}, we find that it is convenient to work with a half-integer parametrization $(p,q)$. That is, for an Eisenstein integer $x_1 + x_2\omega$, we set $p = x_1 - x_2/2$ and $q = x_2/2$. Thus satisfying \Cref{eq:approx-constraints-full-search,eq:exunbounds} corresponds to finding three lattice vectors $(p,q\sqrt{3})^T$ each in a different region. Using $\alpha$ to represent the desired angles and
\begin{equation}
\label{eq:radius}
r_1=3^{\lceil f/2\rceil}\eta(\varepsilon_i),\quad r_2=3^{\lceil f/2\rceil},
\end{equation}
the three regions are then defined by
\begin{align}
    \label{eq:sampling-reduced}
    p \cos\alpha + q\sqrt{3}\sin\alpha \geq r_1 && \text{and} && p^2 + 3q^2 \leq r^2_2.
\end{align}
From the first inequality, it follows that $p^2\cos^2\alpha \geq (r_1 - q\sqrt{3}\sin\alpha)^2$. %
Using the second constraint and completing the square, it can be verified that $|\sqrt{3}q - r_1\sin\alpha| \leq \sqrt{r^2_2 - r^2_1} \cos\alpha$. 
Therefore, the half integers $q$ are sampled within the interval
\begin{align}
    \label{eq:enum-q-bound}
    \frac{1}{2}\left\lceil S_- \right\rceil \leq q \leq \frac{1}{2} \left\lfloor S_+ \right\rfloor,
\end{align}
where $S_\pm := \frac{2}{\sqrt{3}}(r_1 \sin\alpha \pm \sqrt{r^2_2 - r^2_1} \cos\alpha)$. 
For a given $q$, the possible values of $p$ satisfy $p + q \in \mathbb{Z}$ and lie in the interval
\begin{align}
    \label{eq:enum-p-bound}
    \frac{1}{2}\left\lceil 2 T_- \right\rceil \leq p \leq  \frac{1}{2}\left\lfloor 2 T_+ \right\rfloor,
\end{align}
where $T_- := \text{max}\left(-\sqrt{r^2_2 - 3q^2}, \frac{r_1 - \sqrt{3}q\sin\alpha}{\cos\alpha}\right)$ and $T_+ := \sqrt{r^2_2 - 3q^2}$. 

These regions and candidate lattice vectors are depicted on the left of Fig.~\ref{fig:algos}. Then, the diagonal entries of $V$ -- $x_1$, $y_2$ and $z_3$ -- may be easily obtained via division by $(-1)^{\lceil f/2\rceil} \chi^{\bar{f}}$. 

\begin{figure*}[!ht]
   \includegraphics[width=0.9\linewidth]{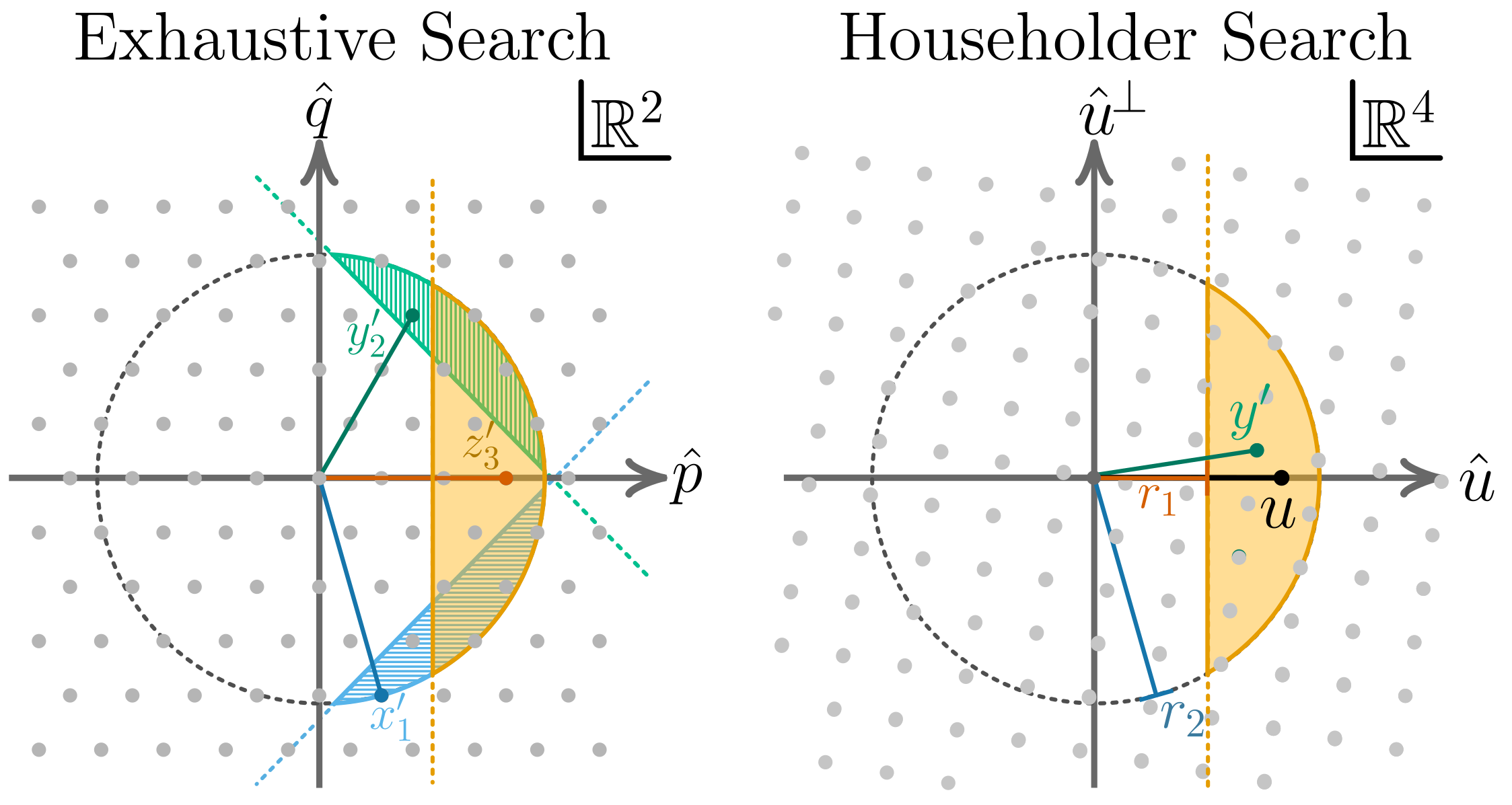}
    \caption{Search regions for: (left) Exhaustive Search Algorithm. The search regions of Eq.~(\ref{eq:approx-constraints-full-search}) are: the blue corresponds to $x^\prime_1$; the green region corresponds to $y^\prime_2$ and the orange region corresponds to $z^\prime_3$; (right) Householder Search Algorithm. $\mathbf{u} = \frac{1}{\sqrt{2}}(\cos\theta/2, \sin\theta/2, -1, 0)^T$. The yellow region corresponds to the search region in Eq.~(\ref{eq:approx-lattice-version}).}
    \label{fig:algos}
\end{figure*}

With these constraints for fixed $f$, one can enumerate possible candidate triplets \( (x_1, y_2, z_3) \). Then for a particular triplet \( (x_1, y_2, z_3)/\chi^f \), one should check that they satisfy \Cref{eq:rznorm} and the necessary and sufficient condition for a set of complex numbers to serve as the diagonal entries of a unitary matrix~\cite{horn1954doubly} which in our case are:
\begin{align}\label{eq:alfred-horn-conditions}
    |x_1| + |y_2| &- |z_3| > \sqrt{3}^f,\,  |x_1| + |z_3| - |y_2| > \sqrt{3}^f,\notag\\
    &|y_2| + |z_3| - |x_1| > \sqrt{3}^f.
\end{align}
For triplets satisfying these conditions, we proceed to solve for the off-diagonal entries.  This unitary matrix completion problem is addressed given a triplet of candidate diagonal entries $(x_1, y_2, z_3)$.
The search for off-diagonal entries can be done by exhaustive search.
The norms of column entries must satisfy the constraint
\begin{align}
    \label{eq:off-diag}
    |x_1|^2+|x_2|^2 + |x_3|^2 &= 3^f
\end{align}
and analogous ones for $y_i,z_i$. To solve these equations, one notes that $|x_2|^2 \leq 3^{f} - |x_1|^2$. 
Thus, for $N=[0,3^f - |x_1|^2]$, one tries to solve $|x_2|^2 = N$ using the method described in App.~\ref{app:norm-eq}. If a solution exists, the process is repeated for $|x_3|^2 = 3^f - |x_1|^2 - |x_2|^2$. One can precompute a lookup table for $|x|^2 = N$ so that this norm equation is solved only once for a given $N$ up to a maximum $N$,
\begin{align}
    N &= \text{max}(3^f - |x_1|^2, 3^f - |y_2|^2, 3^f - |z_3|^2)\notag\\
    &\leq 3^{2\lceil f/2\rceil}\left(\varepsilon^2-\frac{\varepsilon^4}{4}\right),
\end{align}
or in other words, $N = \mathcal{O}(3^f\varepsilon^2)$.
With off-diagonal entries determined, the final step is to verify that the resulting matrix $V$ is unitary. If no valid unitary matrix is found after exhausting all possible triplets, we increment \( f \) and restart the enumeration process. The pseudocode is outlined in Algorithm~\ref{alg:full-approx-algorithm}.

\begin{algorithm}[hbt!]
\caption{Exhaustive search.\phantom{XXXXXXXXX}}
\label{alg:full-approx-algorithm}
\KwData{$\theta$ and $\varepsilon$}
\KwResult{$V\in U\left(3, \mathcal{R}_{3,\chi}\right)$ such that $\|U - V\|_F \leq \varepsilon$}
\(f \gets 0\)

 $\Gamma \gets (-1)^{\lceil f/2 \rceil} \chi^{\bar{f}}$
 
\While{ $V$ is not found}{
    $\mathcal{S}$  $\leftarrow \textsc{EnumerateCandidates}(\theta, f, \varepsilon)$
    
    \For{each $(x^\prime_1, y^\prime_2, z^\prime_3)$ in $\mathcal{S}$ }{
        \If{$\Gamma$ divides $x^\prime_1$, $\Gamma$ divides $y^\prime_2$ \& $\Gamma$ divides $z^\prime_3$}{
            $x_1 \gets x^\prime_1/\Gamma$, $y_2 \gets y^\prime_2/\Gamma$, $z_3 \gets z^\prime_3/\Gamma$
            
            \If {\Cref{eq:alfred-horn-conditions} is satisfied}{
                $V \gets $ \textsc{CompleteUnitary}$(x_1, y_2,z_3)$
                \If{$V$ is found}{
                    \Return $V$ 
                }
            }
        }
    }
    increment $f$
    
    $\Gamma \gets (-1)^{\lceil f/2 \rceil} \chi^{\bar{f}}$
}
\end{algorithm}

The time complexity can be split into two steps. 
\textsc{Step (1)} involves enumerating all valid triplets $(x_1, y_2, z_3)$, where the area of the search regions in Fig.~\ref{fig:algos} determines the complexity. One can show these areas are
\begin{equation}
    3^{2\lceil f/2\rceil}\varepsilon^2\left(1 - \frac{\varepsilon^2}{4}\right)\arccos\left(1 - \frac{\varepsilon^2}{2}\right).
\end{equation}
Expanding $\arccos{\left(1 - \varepsilon^2/2\right)} = \varepsilon + \mathcal{O}(\varepsilon^3)$, we observe that this area scales as $\mathcal{O}(3^f \varepsilon^3)$. Consequently, the total complexity of enumeration is $\mathcal{O}(3^{3f} \varepsilon^9)$.

\textsc{Step (2)} of the algorithm is to complete the unitary. The complexity is dominated by solving Eq.~\eqref{eq:off-diag} for iteratively for $|x|^2 = N$ where $N = \mathcal{O}(3^{f}\varepsilon^2)$.
Which in App.~\ref{app:norm-eq}, we showed leads to a complexity of $\mathcal{O}(3^{f/2} \varepsilon)$. 

Combining the complexity of both parts, we obtain $\mathcal{O}(3^{7f/2}\varepsilon^{10})$ for a fixed $f$. The algorithm terminates at $f_{\text{max}}$ once a solution is found. Recalling that we expect scaling $f_{\text{max}} = c_1 \log\frac{1}{\varepsilon}$, the overall complexity is 
\begin{align}
    \label{eq:full-complexity}
    \sum\limits_{f=0}^{f_{\text{max}}} 3^{\frac{7f}{2}}\varepsilon^{10} &= \mathcal{O}\left(3^{\frac{7}{2}f_{\text{max}}}\varepsilon^{10}\right)= \mathcal{O}\left(\varepsilon^{10 - \frac{7}{2}c_1}\right).
\end{align}

Finally, it is possible to derive a rough lower bound on $c_1$ by estimating the number of Eisenstein integers in each search region as follows. For a lattice $\mathcal{L}$ with basis $B = \left(\mathbf{b}_1, \mathbf{b}_2, {...}, \mathbf{b}_m\right)$, the fundamental domain is
\begin{equation}
    \mathcal{F}(\mathcal{L}) = \left\{\sum_{i}c_i\,\mathbf{b}_i \Big| \, c_i \in \mathbb{R}\, \text{and}\, c_i \in \left[0,1\right)\right\}.
\end{equation}
From Ref.~\cite{Cassels1959}, $\mathcal{L}$ is a uniform tiling of the ambient space with its fundamental domain. As a result, the volume of this domain is
\begin{equation}
\text{vol}(\mathcal{F}(\mathcal{L})) = \sqrt{\det(B^TB)},
\end{equation}
represents the inverse density of the lattice points. Therefore, the number of lattice points within a region, $\mathcal{K}$, of volume $\text{vol}(\mathcal{K})$ can be approximated by the ratio $\frac{\text{vol}(\mathcal{K})}{\text{vol}(\mathcal{F}(\mathcal{L}))}$.

For us, each search regions has a $\text{vol}(\mathcal{K})=\mathcal{O}(3^f\varepsilon^3)$, and $\mathcal{L}_1$, we find $\text{vol}(\mathcal{F}(\mathcal{L}_1)) = \sqrt{3}/2$. Requiring that at least one lattice point exists per volume then corresponds to 
\begin{equation}
    1\leq\frac{\text{vol}(\mathcal{K})}{\text{vol}(\mathcal{F}(\mathcal{L}))}=\frac{3^f\varepsilon^3}{\sqrt{3}/2}=2\times3^{c_1 \log_3\frac{1}{\varepsilon}-\frac{1}{2}}\varepsilon^3.
\end{equation} 
Reducing this inequality yields $c_1 \geq 3$ and consequently $N_{\mathbf{R}} \geq 3\log_3\frac{1}{\varepsilon} + C$ where $C$ is a constant.

\section{Householder reflection search}\label{sec:algorithm-two}

The poor scaling of the time complexity with $\varepsilon$ of Algorithm~\ref{alg:full-approx-algorithm} motivates us to search for more efficient ways to approximate diagonal gates.  It was argued in Ref.~\cite{PhysRevA.93.012313} that by restricting to  Householder reflection of the form 
\begin{equation}
    R_{\mathbf{u}} = \mathbb{1} - 2 \mathbf{u} \mathbf{u}^\dagger
\end{equation} where $\mathbf{u}$ is a unit vector  there exists a probabilistic classical algorithm returning an approximation with average $N_R \sim 8\log_3\frac{1}{\varepsilon}$ using an average classical complexity $\mathcal{O}(\log\frac{1}{\varepsilon})$. Further work in Ref.~\cite{Bocharov:2018zfk} reduced these estimates under certain number-theoretic conjectures to $N_R \sim 5\log_3\frac{1}{\varepsilon}$ for `non-exceptional' target two-level unit vectors.  We present in these works an explicit algorithm and demonstrate its scaling for Householder reflections in the $(\mathbf{C} + \mathbf{R})_3$ group. 

Ref.~\cite{PhysRevA.93.012313} reformulated the approximation problem to only require approximating a unit vector $\mathbf{u}=(u_1,u_2,u_3)$ with another unit vector   $\mathbf{v} = \frac{1}{\chi^f}\left(v_1, v_2, v_3\right)^T$ such that the entries $v_{i} \in \mathcal{R}_3$ and $f \in \mathbb{N}_0$.\footnote{This problem is related to intrinsic Diophantine approximation~\cite{Cha2023}.} 
We refer to such a unit vector as a unit Eisenstein vector. It is straightforward to verify that a $R^Z_{(0,1)}(\theta)$ matrix, up to a permutation, corresponds to a Householder reflection i.e.  $R^Z_{(0,1)}(\theta) = X_{(0,1)} R_{\mathbf{u}}$ for some 
\begin{equation}
\label{eq:houseuvec}
\mathbf{u}=\frac{1}{\sqrt{2}}(e^{\mathrm{i}\theta/2}, -1, 0)^T.
\end{equation}

We can derive an upper bound between two Householder reflections $R_\mathbf{u}$ and $R_\mathbf{v}$. To start, the norm between two Householder reflections can be related to their associated unit vectors
\begin{equation}
    \|R_\mathbf{u} - R_\mathbf{v}\|^2 = 8\left(1 - \left|\mathbf{u}^\dagger\mathbf{v}\right|^2\right).
\end{equation}
Further, for any two vectors $\mathbf{u}$ and $\mathbf{v}$
\begin{equation}
    \|\mathbf{u} - \mathbf{v}\|^2 = 2\left(1 - \operatorname{Re}\left(\mathbf{u}^\dagger\mathbf{v}\right)\right).
\end{equation}
Since for any complex number $z$ the following inequality is true $(\text{Re}\left[z\right])^2 \leq |z|^2$, combining with the previous expressions yields
\begin{align}
    \label{eq:householder-error-bound}
    \|R_{\mathbf{u}} - R_{\mathbf{v}}\| \leq 2 \sqrt{2}\,\|\mathbf{u} - \mathbf{v}\| \, \delta\left(\mathbf{u}, \mathbf{v}\right), 
\end{align}
where 
\begin{equation}
    \delta\left(\mathbf{u}, \mathbf{v}\right) : = \sqrt{1 - \frac{\|\mathbf{u} - \mathbf{v}\|^2}{4}}\leq 1.
\end{equation} 
Using $\delta(\mathbf{u}, \mathbf{v})\leq1$ reproduces the bounds of Ref.~\cite{Kliuchnikov:2013ela} which demonstrates to approximate $R_{\mathbf{u}}$ within  $\varepsilon$, it suffices to identify a $\mathbf{v}$ such that
\begin{equation}
    ||\mathbf{u}-\mathbf{v}||\leq \varepsilon/(2\sqrt{2}).
\end{equation}
However, given $\delta(\mathbf{u}, \mathbf{v}) \leq 1$ in~\Cref{eq:householder-error-bound} means that the requirement $\|\mathbf{u} - \mathbf{v}\| \leq \varepsilon/(2\sqrt{2})$ may exclude some $R_{\mathbf{v}}$ that meet the desired accuracy. To address this, we introduce into our algorithm a contraction factor $0 < c \leq 1$ and adjust the tolerance to $\varepsilon/(2\sqrt{2} \,c)$. While $c < 1$ no longer guarantees that all reflections remain within $\varepsilon$, those exceeding this threshold can be checked and rejected. This approach reduces $N_{\mathbf{R}}$ without increasing the algorithmic complexity, but the enlarger search space will lead in practice to longer run times.

From \Cref{eq:houseuvec}, we see that $u_3 = 0$. Thus approximating $\mathbf{u}$ by $\mathbf{v} = \frac{1}{\sqrt{-3}^f}\left(v_1, v_2, v_3\right)^T$ imposes the condition

\begin{align}
    \label{eq:approx-cond-relabelled}
    \operatorname{Re}&\left(u_1^* v^\prime_1 + u_2^* v^\prime_2 \right) \geq r_1,
\end{align}
where we have rewritten terms using \Cref{eq:cov,eq:radius}. Additionally, since  $\mathbf{v}$ is a unit vector, it follows that
\begin{equation}
\label{eq:house_norm}
    |v^\prime_1| + |v^\prime_2|^2 \leq r^2_2.
\end{equation}

These constraints on $v^\prime_1$ and $v^\prime_2$ have a straightforward formulation and geometric interpretation in $\mathbb{R}^4$. Any two-level complex vector $ \mathbf{u}  \in \mathbb{C}^3$ (with $u_3=0$ in our case) can be mapped to $\mathbb{R}^4$ by the isomorphism
    $\mathbf{u} = (u_1, u_2, 0)^T  \mapsto \mathbf{u} = \left(\operatorname{Re}(u_1), \operatorname{Im}(u_1), \operatorname{Re}(u_2), \operatorname{Im}(u_2)\right)^T$. By a slight abuse of notation, we refer to both the two-level complex and its image in $\mathbb{R}^4$ using the same letter. When the context does not make it obvious which version is being discussed, we will clarify it explicitly.

In addition, consider for now the complex vector $(v_1, v_2)^T \in \mathbb{C}^2$ where $v_1, v_2 \in \mathcal{R}_3$. By defining these Eisenstein integers as $v_1 = x_1 + x_2 \omega$ and $v_2 = x_3 + x_4\omega$ where $x_i \in \mathbb{Z}$, we can represent the complex vector $(v_1, v_2)^T$ in $\mathbb{R}^4$ by a vector $\mathbf{y} = \left(x_1 - \frac{x_2}{2}, \frac{\sqrt{3}}{2} x_2, x_3 - \frac{x_4}{2}, \frac{\sqrt{3}}{2} x_4\right)^T$. Consequently, it can be seen that the set of all such vectors $\mathbf{y}$ forms an integer lattice $\mathcal{L}_2$ in $\mathbb{R}^4$, defined by
\begin{align}
    \label{eq:lattice-vector}
    \mathcal{L}_2 = \left\{\mathbf{y} = B_2\,\mathbf{x} \Big| x \in \mathbb{Z}^4 \right\} ,
\end{align}
where $B_2 = B_1 \oplus B_1$ and $B_1$ is defined in Eq.~\eqref{eq:eis-lattice-basis}. The Eisenstein integers $v_1$ and $v_2$ can be recovered from $\mathbf{y}$ via $\mathbf{x} = B^{-1}_2 \mathbf{y}$. 

To finally formulate the conditions in \Cref{eq:approx-cond-relabelled,eq:house_norm} as vector constraints in $\mathbb{R}^4$, we denote the image of $(v^\prime_1, v^\prime_2)^T$ by a lattice vector $\mathbf{y}^\prime \in \mathcal{L}_2$. Then, these constraints can be rewritten as lattice points $\mathbf{y^\prime}$ such that
\begin{align}
    \label{eq:approx-lattice-version}
    \mathbf{u}^T \mathbf{y}^\prime \geq r_1  && \text{and} && \mathbf{y^\prime}^{T} \mathbf{y^\prime} \leq r^2_2,
\end{align}
 In other words, we are looking for lattice points $\mathbf{y}^\prime$ above a hyperplane and inside a hypersphere of radius $r_2$ as demonstrated in Fig.~\ref{fig:algos}. 
Denoting the volume between the hyperplane ($\mathbf{u}^T \mathbf{y}' = r_1$) and the hypersphere by $\mathcal{D}(\mathbf{u}, f, \varepsilon)$, we have $\mathbf{y}^\prime \in \mathcal{L}_2 \cap \mathcal{D}(\mathbf{u}, f, \varepsilon)$. 

Once such an Eisenstein vector $\mathbf{v}=\frac{1}{\chi^f}\left(v_1, v_2, v_3\right)^T$ is found, the matrix $V = X_{(0,1)} R_{\mathbf{v}}$ is the desired approximation of $R^Z_{(0,1)}(\theta)$. It can be shown that synthesizing $V$ requires $N_\mathbf{R}\leq2f$~\cite{Kalra:2023llu}, and we will provide a deterministic search algorithm focused on minimizing $f$.

Approximating $V$ can thus be done in two iterated steps. \textsc{Step (1)}: Setting $\varepsilon^\prime = \varepsilon/(2 \sqrt{2} c)$, for a fixed $f \in \mathbb{N}_0$ and a unit vector $\mathbf{u} \in \mathbb{R}^4$, we enumerate all candidates $\mathbf{y}^\prime \in \mathcal{L}_2 \cap  \mathcal{D}(\mathbf{u}, f, \varepsilon^\prime)$ following the algorithm in App.~\ref{app:enumeration}. 
\textsc{Step (2)}: For each $\mathbf{y^\prime}$, we extract $v_1$ and $v_2$, and solve the norm equation $|v_3|^2 = 3^f - |v_1|^2 - |v_2|^2$ as discussed in App.~\ref{app:norm-eq}. 
If a solution is found, we construct the vector $\mathbf{v}$ and subsequently the matrix $V = X_{(0,1)}R_{\mathbf{v}}$. It only remains to test that $V$ approximates the target $R^Z_{(0,1)}(\theta)$ to the desired accuracy. Therefore, if $\|R^Z_{(0,1)}(\theta) - V\| \leq \varepsilon$, the matrix $V$ is returned. In the event that all candidates $\mathbf{y}^\prime$ are exhausted and no solution is found, we increment $f$ and repeat the procedure. The details are summarized in Algorithm~\ref{alg:approximation-algorithm}.

\begin{algorithm}
\caption{Householder reflection search.\phantom{XXX}}
\label{alg:approximation-algorithm}
 \KwData{$\theta$, $\varepsilon$, and $0 < c \leq 1$}
 \KwResult{$V \in (\mathbf{C} + \mathbf{R})_3$ group with $\|R^Z_{(0,1)}(\theta) - V\| \leq \varepsilon$.}
 
 $\mathbf{u} \gets \frac{1}{\sqrt{2}}(\cos\theta/2, \sin\theta/2, -1, 0)^T$
 
 $\varepsilon^\prime \gets \varepsilon/(2 \sqrt{2} c)$ 
 
 \(f \gets 0\)
 
 $\Gamma \gets (-1)^{\lceil f/2 \rceil} \chi^{\bar{f}}$
 
\While{$\mathbf{v}$ is not found}{
    \For{ \text{each } $\mathbf{y}^\prime$ in $\mathcal{L}_2 \cap \mathcal{D}(\mathbf{u}, f, \varepsilon^\prime)$ }{
         $\mathbf{x^\prime} \gets B^{-1}_2 \mathbf{y^\prime}$
         
         $v^\prime_1 \gets x^\prime_1 + x^\prime_2\omega$ and $v^\prime_2 \gets x^\prime_3 + x^\prime_4\omega$
         
        \If{$\Gamma$ divides $v^\prime_1$ and $\Gamma$ divides $v^\prime_2$}{
             $v_1 \gets v^\prime_1/\Gamma$ and $v_2 \gets v^\prime_2/\Gamma$
             
             solve $|v_3|^2 = 3^f - |v_1|^2 - |v_2|^2$ 
             
            \If{\text{$v_3$ is found}}{
                 $\mathbf{v} \gets \frac{1}{\chi^f}(v_1, v_2, v_3)^T$
                 
                 $V \gets X_{(0,1)} R_{\mathbf{v}}$
                 
                \If{$\|R^Z_{(0,1)}(\theta) - V\| \leq \varepsilon$}{
                    \Return $V$ 
               }
               }
               }
               }
     increment $f$
     
     $\Gamma \gets (-1)^{\lceil f/2 \rceil} \chi^{\bar{f}}$
}
\end{algorithm}

We now discuss the complexity of this algorithm. In \textsc{step (1)}, we enumerate all $\mathbf{y}'$ in the region bound by \Cref{eq:approx-lattice-version}. For a fixed $f$, the complexity of this corresponds to the volume of the region. Due to spherical symmetry, this volume corresponds to that of a hyperspherical cap in $\mathbb{R}^4$ defined by the constraints $y_4 \geq r_1$
and $\mathbf{y}^T\mathbf{y} \leq r^2_2$. According to Ref.~\cite{Li2010}, this volume is given by:
\begin{align}
    \label{eq:vol-4dim-cap1}
    \operatorname{vol}( \mathcal{D}(\mathbf{u}, f, \varepsilon)) &= \frac{\pi^{3/2}}{\Gamma(5/2)}r^{4}_2\int_{0}^{\phi} \sin^4\theta \, d\theta,
\end{align}
where $\phi = \arccos(r_1/r_2) = \arccos(1 - \frac{\varepsilon^2}{2})$. Evaluating this integral yields:
\begin{align}
    \label{eq:vol-4dim-cap2}
    \operatorname{vol}( \mathcal{D}(\mathbf{u}, f, \varepsilon)) &= \frac{\pi r^{4}_2}{24} \left(12\phi - 8\sin{(2\phi)} + \sin{(4\phi)}\right)\notag\\
    &=\frac{16\sqrt{2}\pi}{15}r_2^4\varepsilon^5+\mathcal{O}(r^4_2\varepsilon^7),
\end{align}
where we expanded in the $\varepsilon\rightarrow 0$ limit. Using the definition of $r_2$ in \Cref{eq:radius} gives us a final scaling for \textsc{Step (1)} of $\mathcal{O}\left(3^{2f} \varepsilon^5\right)$.

 The problem in \textsc{step (2)} is to solve the norm equation in \Cref{eq:house_norm} using the exhaustive search method in App.~\ref{app:norm-eq} has a complexity $\mathcal{O}(|v_3|)$. We show in App.~\ref{app:v3--normbound} that $|v_3| = \mathcal{O}\left(3^{f/2}\varepsilon\right)$. Combining these two steps, the complexity of the Householder search method is $\mathcal{O}(3^{5f/2} \varepsilon^6)$. Similarly to the exhaustive algorithm, this one terminates at $f_{\text{max}}$. Assuming the scaling $f_{\text{max}} = c_2\log \frac{1}{\varepsilon}$, the overall 
\begin{align}
    \label{eq:house-complexity}
    \sum\limits_{f=0}^{f_{\text{max}}} 3^{\frac{5}{2}f}\varepsilon^{6} &= \mathcal{O}\left(3^{\frac{5}{2}f_{\text{max}}}\varepsilon^{6}\right)= \mathcal{O}\left(\varepsilon^{6-\frac{5}{2}c_2}\right).
\end{align}
Similarly, it is also possible to place a rough lower bound on $c_2$. The volume of the search region in Fig.~\ref{fig:algos} (right) is $\mathcal{O}(3^{2f}\varepsilon^5)$. Using the arguments at the end of \sref{sec:algorithm-one}, we obtain $c_2 \geq 2.5$. This lower bound corresponds to $N_{\mathbf{R}} \geq 5 \log_3\frac{1}{\varepsilon} + C$ where $C$ is a constant.

\section{Exact Synthesis}\label{sec:exact-synthesis}
Given the approximation gate $V \in U(3,\mathcal{R}_{3,\chi})$, what remains is to determine the word in $(\mathbf{C}+\mathbf{T})_3$ which produces it.  To do this, one using the fact that any $V$ can be written optimally in a normal form~\cite{Kalra:2023llu}:
\begin{equation}
\label{eq:normal_form}
    V=\prod_i^f HD(a_{0,i}, a_{1,i}, a_{2,i})\mathbf{R}^{\varepsilon_i} X^{\delta_i}
\end{equation}
where $a_{0,i}, a_{1,i}, a_{2,i},\delta_i \in \{0, 1, 2\}$, $\varepsilon_i \in \{0, 1\}$. While this form is more complicated than the analogous normal form for qubits~\cite{matsumoto2008representationquantumcircuitsclifford} which contains only $H,T,S$, there is still only one non-Clifford gate, $\mathbf{R}$ since $D(a, b, c)$ and $X$ can be constructed from $H,S$ as shown above. Taken together, the total number of possible normal forms $N_N$ smaller than $f$ is \begin{equation}
    N_N=\sum_{i=0}^f(3^42^2)^i=\frac{324^{f+1}-1}{323}.
\end{equation}
To determine the correct normal form, and therefore circuit, we implement the algorithm from Ref.~\cite{Kalra:2023llu}, which guarantees optimality in $N_\mathbf{R}$.

To do this, we take advantage of the \emph{smallest denominator exponent} $\text{sde}(z)$ which corresponds to the smallest non-negative integer $f$ such that $z\chi^f\in \mathcal{R}_{3}$ when $z\in\mathcal{R}_{3,\chi}$. It was shown in Ref.~\cite{Kalra:2023llu} that all entries of a matrix in $U(3,\mathcal{R}_{3,\chi})$ have the same $\text{sde}(z)$. Therefore an algorithm for solving \Cref{eq:normal_form} is to iteratively by finding a set $a_{0,i}, a_{1,i}, a_{2,i},\delta_i,\varepsilon_i$ which reduces the $\text{sde}(z)$ by 1. Each reduction corresponds to most one $\mathbf{R}$ gate per iteration. Once the $\text{sde}(z)=0$, the resulting unitary is a Clifford element, which can then be obtained from a lookup table. The pseudocode is detailed in Algorithm~\ref{alg:unitary_decomposition}.

\begin{algorithm}
\caption{Decomposition of $V$ in $U(3, \mathcal{R}_{3,\chi})$.}
\label{alg:unitary_decomposition}
 \KwData{Unitary $U$  in  $U(3, \mathcal{R}_{3,\chi})$.\\
$T_{D}$ - table of all zero-sde unitaries in $U(3, \mathcal{R}_{3,\chi})$.}
 \KwResult{Sequence $S_{\text{out}}$ of $H, D, \mathbf{R}$, and $X$ gates that implement $U$.}
 $u \gets$ top left entry of $U$
 
 $S_{\text{out}} \gets$ Empty
 
 $s \gets \mathrm{sde}(u)$

\While{$s > 0$}{
     state $\gets$ unfound

    \ForAll{$a_0, a_1, a_2, \delta \in \{0, 1, 2\}, \varepsilon \in \{0, 1\}$}{
        \While{state = unfound}{
             $u \gets$ $(HD(a_0, a_1, a_2)\mathbf{R}^{\varepsilon}X^\delta U)_{00}$
             
            \If{$\mathrm{sde}(u) = s - 1$}{
                 state $\gets$ found
                 
                 append $X^{-\delta}\mathbf{R}^{-\varepsilon}D^{-1}H$ to $S_{\text{out}}$
                 
                 $s \gets \mathrm{sde}(u)$
                 
                 $U \gets HD(a_0, a_1, a_2)\mathbf{R}^{\varepsilon}X^\delta U$
            }
            }
            }
            }
 lookup matrix $S_{\text{rem}}$ for $U$ in $T_D$
 
 append $S_{\text{rem}}$ to $S_{\text{out}}$
 
 \textbf{return} $S_{\text{out}}$
\end{algorithm}

\section{Numerical Results}\label{sec:numerical-results}
We provide the scaling for the $N_{\mathbf{R}}$ compared to the infidelity, $\varepsilon$ in Fig.~\ref{fig:num_results}.  Using the exhaustive algorithm, we evaluated 100 randomly sampled angles in the interval $(-\pi/2,\pi/2)$ at 10 target precisions $\varepsilon\in\{1,0.5,0.25,0.1\hdots10^{-3}\}$. From this, we find on average
\begin{align}
\label{eq:numexh}
    N^E_{\mathbf{R}}(\varepsilon) &= \yintm + \slopem \text{log}_{10}(1 / \varepsilon)\notag\\
    &= \yintm + 4.113(3)\text{log}_{3}(1 / \varepsilon).
\end{align} For the Householder search algorithm we again evaluated 100 angles  at 10 target precisions $\varepsilon\in\{1,10^{-1}\hdots10^{-9}\}$ and 36 angles at $10^{-10}$. After a modest search through contraction factors, we found that $c=0.35$ yields short word lengths with tolerable additional runtime.  Using this faster algorithm gave
\begin{align}
\label{eq:numhos}
N^H_{\mathbf{R}}(\varepsilon) &= \yint + \slope \text{log}_{10}(1 / \varepsilon)\notag\\
&= \yint +  5.139(14)\text{log}_{3}(1 / \varepsilon).
\end{align}
We can use the average costs (Eqs. (\ref{eq:numhos}) and (\ref{eq:numexh})), along with Eqs. (\ref{eq:full-complexity}) and (\ref{eq:house-complexity}), to compute the average complexity of both algorithms. For the full $(\mathbf{C} + \mathbf{R})_3$ search, the average $c_1$ corresponds to the slope $\pgfmathprintnumber[fixed,precision=2]{\slopemthree}$, which yields an average complexity of $\mathcal{O}(\varepsilon^{\pgfmathprintnumber[fixed,precision=2]{\fullcomplexity}})$. In the case of the Householder search, the average $c_2$ is half the slope, $\pgfmathprintnumber[fixed,precision=3]{\slopethree}$, resulting in an average complexity of $\mathcal{O}(\varepsilon^{\pgfmathprintnumber[fixed,precision=2]{\householdercomplexity}})$.

In order to determine the average cost for an arbitrary single-qutrit SU(3) gate, these results should be multiplied by 6~\cite{Bocharov:2018zfk}. Given the importance of the $\mathbf{T}_3$ gate in discussions of qutrits, we investigate its approximate synthesis with the Householder search algorithm, finding $N_{R}^H$ for $\mathbf{T}_3$ aligns with that of the average gate costs. In passing, we also note that $N_{\mathbf{R}}$ exhibited a weak angular dependence: angles closer to $R_{3,\chi}$ lead to lower $N_\mathbf{R}$. 

\begin{figure}
    \centering
    \includegraphics[width=\linewidth]{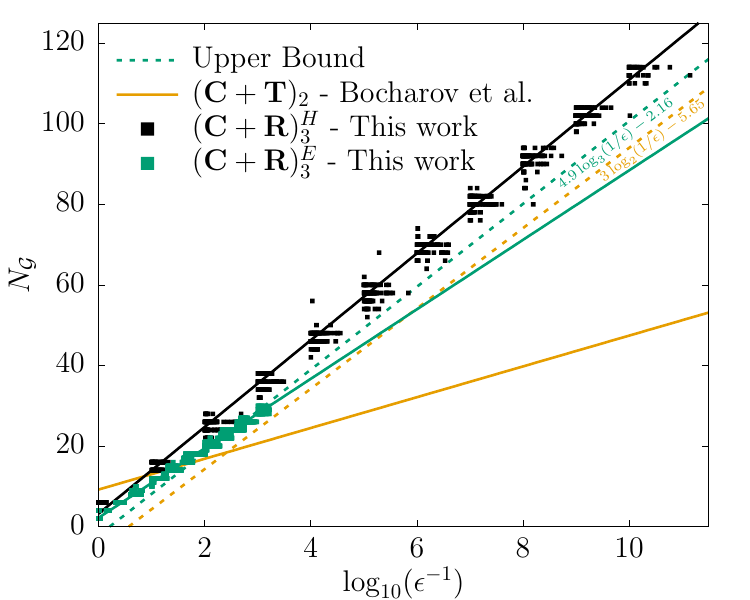}
    \caption{Scaling of the number of non-Clifford gates, $N_{\mathcal{G}}$ ($\mathbf{T}$ or $\mathbf{R}$ depending on qudit dimension) against the infidelity $\varepsilon$. Angles are uniformly sampled in the region $\theta\in(-\pi/2
     , \pi/2)$.}
    \label{fig:num_results}
\end{figure}

With these results, we can compare our single-qutrit synthesis method to that of implementing the same unitary on two qubits. For the general case of two-qubit circuits with CNOTs and $R^{\alpha}(\theta)$, i.e. $SU(4)$, the number of 15 single-qubit rotations is required~\cite{Shende:2004gqq}.  Restricting to the single-qutrit subspace of $SU(3)$, dimensional analysis bounds the cost as at least 10 $R^{\alpha}(\theta)$.  Using the average cost for synthesizing $R^Z(\theta)$ from \cite{PhysRevLett.114.080502} of $N^{RUS}_T=9.2+3.817\log_{10}(1/\varepsilon)=9.2+1.15\log_2(1/\varepsilon)$ would estimate that an arbitrary single-qutrit gate would require at least $10~N^{RUS}_T$.  Comparing to the costs in Eqs.~(\ref{eq:numexh}) and (\ref{eq:numhos}) imply that single-qutrit synthesis via our algorithms for $\mathbf{R}$ incurs an overhead factor as $\varepsilon\rightarrow 0$ of 1.35 and 1.69 respectively. If one instead considers a fiducial $\varepsilon=10^{-10}$, these factors reduce to 1.12 and 1.40.

\section{Conclusion}
\label{sec:conclusion}
In this work, we have demonstrated two algorithms to synthesize diagonal gates for qutrits using the Clifford + $\mathbf{R}$ gates. 
Our studies show that given a target infidelity $\varepsilon$ for a diagonal rotation gate, one can approximate a diagonal that requires approximately $\yint + \slope \text{log}_{10}(1 / \varepsilon)$ $\mathbf{R}$ gates.
These results open up the feasibility of using fault-tolerant qutrits for quantum simulations. While these results are valuable, they leave several open questions. First, while the prefactor $\slope$ for synthesizing diagonal gates is close to the lower bound of 10.27, multiplying by six to obtain arbitrary gates leaves us quite far from optimality.  Potential improvements could come from exploring repeat-until-success methods~\cite{PhysRevLett.114.080502} or identifying broader subclasses of gates that can be efficiently synthesized.  Another direction of research might investigate synthesis with other groups, such as $(\mathbf{C}+\mathbf{D})_3$ or $(\mathbf{C}+\mathbf{T})_3$.  Furthermore, similar to qubits, looking for larger transverse groups than Cliffords could further reduce the cost~\cite{Evra_2022,Kubischta:2024fuf} and potentially enable novel application-specific gate sets e.g. high energy physics~\cite{Kurkcuoglu:2024cfv,Gustafson:2024kym}.

\begin{acknowledgments}
The authors thank Namit Anand, Di Fang, Yao Ji, Aaron Lott, and Yu Tong for their helpful comments and suggestions. This material is based on work supported by the U.S. Department of Energy, Office of Science, National Quantum Information Science Research Centers, Superconducting Quantum Materials and Systems Center (SQMS) under contract number DE-AC02-07CH11359. This work was produced by Fermi Forward Discovery Group, LLC under Contract No. 89243024CSC000002 with the U.S. Department of Energy, Office of Science, Office of High Energy Physics. EG was supported by the NASA Academic Mission Services, Contract No. NNA16BD14C and the Intelligent Systems Research and Development-3 (ISRDS-3) Contract 80ARC020D0010 under the NASA-DOE interagency agreement SAA2-403602. EM acknowledges support from U.S. Department of Energy, Office of Science, Nuclear Physics, under the grant number DE-FG02-95ER40907. SZ acknowledges support from the U.S. Department of Energy, Office of Science, Accelerated Research in Quantum Computing Centers, Quantum Utility through Advanced Computational Quantum Algorithms, grant No. DE-SC0025572. DL acknowledges support from DMREF Award No. 1922165, Simons Targeted Grant Award No. 896630 and Willard Miller Jr. Fellowship.
\end{acknowledgments}

\bibliography{bibo}

\appendix 

\section{Enumeration algorithm for \texorpdfstring{$\mathbf{y} \in \mathcal{L}_2 \cap   \mathcal{D}(\mathbf{u}, f, \varepsilon)$}{}}\label{app:enumeration}
The goal is to enumerate all vectors $\mathbf{y} \in \mathcal{L}_2 \cap \mathcal{D}(\mathbf{u}, f, \varepsilon)$ for a fixed $f \geq 0$. As defined in \Cref{eq:lattice-vector}, the lattice vectors $\mathbf{y} \in \mathcal{L}_2$ take the form $\mathbf{y} = B_2 \mathbf{x}$, where $\mathbf{x} \in \mathbb{Z}^4$. By matrix multiplication, $\mathbf{y}$ can be written as
\begin{equation}
    \mathbf{y} = \left(x_1 - \frac{x_2}{2},\frac{\sqrt{3}}{2} x_2, x_3 - \frac{x_4}{2}, \frac{\sqrt{3}}{2} x_4\right)^T.
\end{equation}
It is convenient to parametrize $\mathbf{y}$ with half-integers $p_1, q_1, p_2, q_2$ such that $p_i := x_i - x_{2i}/2$ and $q_i := x_{2i}/2$ where $i = 1,2$. 
Each pair $(p_i, q_i)$ must satisfy the integer constraint $x_i = p_i + q_i$. 
Consequently, the lattice vectors take the form $\mathbf{y} = \left(p_1, \sqrt{3}q_1, p_2, \sqrt{3}q_2\right)^T$ and $\mathbf{y}^T \mathbf{y} = \sum\limits_{i=1}^{2} p^2_i + 3q^2_i$.

\begin{lemma}\label{lemma:y4-bound}
    For $\mathbf{y} \in  \mathcal{D}(\mathbf{u}, f, \varepsilon)$, $|y_4| \leq \sqrt{r^2_2 - r^2_1}$.
\end{lemma}
\begin{proof}
    Let $\mathbf{e}_i$ (with $i = 1,2,3,4$) be the canonical basis vectors of $\mathbb{R}^4$. 
    That is, e.g. $\mathbf{e}_1 = (1,0,0,0)^T$. Let $\Pi_3$ be the projector onto the subspace spanned by $\mathbf{e}_1, \mathbf{e}_2$ and $\mathbf{e}_3$. 
    Since $u_4 = 0$,  it follows that $\mathbf{u}^T \mathbf{y} = \mathbf{u}^T \Pi_3 \mathbf{y}$.
    By the triangle inequality, $(\Pi_3 \mathbf{y})^T (\Pi_3 \mathbf{y})  \geq r^2_1$. 
    Additionally, using the total norm constraint: $\mathbf{y}^T \mathbf{y} = (\Pi_3 \mathbf{y})^T (\Pi_3 \mathbf{y}) + y^2_4 \leq r^2_2$, implying that $y^2_4 \leq r^2_2 - r^2_1$ or equivalently, $|y_4| \leq \sqrt{r^2_2 - r^2_1}$. 
\end{proof}
This lemma constrains the sampling range for $q_2$:
\begin{align}
    \frac{1}{2}\left\lceil -2\, \sqrt{\frac{r^2_2 - r^2_1}{3}} \, \right\rceil \leq q_2 \leq \frac{1}{2}\left\lfloor 2\, \sqrt{\frac{r^2_2 - r^2_1}{3}} \, \right\rfloor.
\end{align}

For each $q_2$, the other components satisfy:
\begin{align}
    \label{eq:enum-p2}
    q_1 \cos(\alpha)  + \sqrt{3}q_1\sin(\alpha) &\geq \sqrt{2} r_1 + p_2,\notag\\
    q^2_1 + 3q^2_1 &\leq r^2_2 - 3q^2_2 - p^2_2
\end{align}
where $\alpha := \theta/2$. 
Since $(\cos\alpha, \sin\alpha)^T$ is a unit vector, applying the arguments in Lemma~\ref{lemma:y4-bound} gives:
\begin{align}
    \left|\sqrt{2}r_1 + p_2 \right| \leq \sqrt{r^2_2 - 3q^2_2 - p^2_2}.
\end{align}
This inequality holds only for $p_2$ within the interval:
\begin{align}
    \frac{1}{2}\left\lceil 2 \, p_{2,\text{min}} \right\rceil \leq p_2 \leq \frac{1}{2}\left\lfloor 2 \, p_{2,\text{max}} \right\rfloor,
\end{align}
where $ p_{2,(\text{max}, \text{min})} := \frac{1}{\sqrt{2}}\left(-r_1 \pm \sqrt{r^2_2 - r^2_1 - 3q^2_2}\right)$. Note that ensuring that the radicand is nonnegative provides an elementary proof of Lemma~\ref{lemma:y4-bound}.
With $p_2$ and $q_2$ determined, Eq.~\eqref{eq:enum-p2} can be rewritten as:
\begin{align}
    \label{eq:q1-p1}
    p_1 \cos(\alpha)  + \sqrt{3}q_1\sin(\alpha) &\geq \lambda_1, 
    \quad p^2_1 + 3q^2_1 \leq \lambda^2_2,
\end{align}
with $\lambda_1 := \sqrt{2} r_1 + p_2$ and $\lambda^2_2 := r^2_2 - 3q^2_2 - p^2_2$. Enumerating such $p_1$ and $q_1$ can be done using the results in \Cref{eq:enum-q-bound,eq:enum-p-bound}.

\section{Solving the norm equation}
\label{app:norm-eq}

This section describes how to solve the norm equation. Specifically, given a positive integer $N$, the problem is to find an Eisenstein integer $x = a + b\omega$ such that $|x|^2 = N$. Using the half-integer representation, $p = a - \frac{b}{2}, \; q = \frac{b}{2}$, the norm equation can be rearranged to:
\begin{align}
    p^2 + 3q^2 - N = 0. 
\end{align}
Interpreting this as a quadratic equation in $p$, the discriminant is given by $\Delta = 4N - 12q^2$. 
Real solutions for $p$ exist if $\Delta \geq 0$, which requires $q^2 \leq \frac{N}{3}$. 
Hence, valid values of $q$ are half-integers satisfying $|q| \leq \lfloor \sqrt{\frac{N}{3}} \rfloor$. 

Since the equation is symmetric under $q \to -q$, it suffices to consider only non-negative values of $q$. For each $q$, if $p = \pm\sqrt{N - 3q^2}$ such that $p + q \in \mathbb{Z}$, then a valid solution exists. Together, this gives an $x = (p+q) + (2q)\omega$ which satisfies the norm equation.
Finally, the complexity of this search is $\mathcal{O}(\sqrt{N})$.

It is important to note that because the norm of Eisenstein integers is multiplicative
\begin{equation}
    \left|\prod\limits_i x_i\right|^2 = \prod_i\left|x_i\right|^2\text{ for }x_i \in \mathcal{R}_3,
\end{equation} the norm equation $|x|^2 = N$ for $ \in \mathcal{R}_{3}$ can be solved with a factoring algorithm. Given an integer factorization $N = \prod_i p^{c_i}_i$, there is a method to solve $|x|^2 = N$ which relies on the relation between rational primes and Eisenstein primes, see e.g. Ref.~\cite{6483196}. Indeed, a rational prime $p \neq 3$ either remains prime in $\mathcal{R}_3$ or splits in $\mathcal{R}_3$. That is, if $p \in \mathbb{Z}$ is prime, then $p$ is also prime in $\mathcal{R}_3$ if $p \equiv 2 \, (\text{mod} \, 3)$. On the other hand, if $p \equiv 1\,(\text{mod} \,3)$, there exists $\eta \in \mathcal{R}_3$ such that $|\eta|^2 = p$. In the case $p = 3$,  $|1 - \omega|^2 = 3$.

Then, having the factorization for $N$, solving $|x|^2 = N$ reduces to solving each $|x_i|^2 = p_i$ for the case $p_i \equiv 1 \, (\text{mod} \, 3) $\footnote{For $p_i = 3$, $x_i = 1 - \omega$. For $p_i \equiv 2 \, \text{mod}\, 3$, $x_i = p_i^{c_i/2}$ if and only if $c_i$ is even because such $p_i$ is prime in $\mathcal{R}_3$.}. This can be solved using the method described earlier with complexity only $\mathcal{O}(\sqrt{p_i})$. However, integer factorization using a General Number Field Sieve (GNFS) algorithm~\cite{Buhler1993} would reduce the complexity. 

\section{Bound on \texorpdfstring{$|v_3|$}{}} \label{app:v3--normbound}
As shown in App.~\ref{app:norm-eq}, the complexity of finding $v_3$ via exhaustive search depends on $|v_3|$. This section establishes an upper bound on this value.
For all $\mathbf{y^\prime} \in \mathcal{L}_2 \cap \mathcal{D}(\mathbf{u}, f, \varepsilon)$, the triangle inequality implies $|\mathbf{u}^T \mathbf{y^\prime}| \leq |\mathbf{y^\prime}|$, which can be used to show $\mathbf{y^\prime}^T \mathbf{y^\prime} \geq r^2_1$.

For Eisenstein $\mathbf{y}^\prime$, it follows that $\mathbf{y^\prime}^{T} \mathbf{y}^\prime = |v^\prime_1|^2 + |v^\prime_2|^2 \leq r^2_2$. 
Applying the variable change of \Cref{eq:cov} we obtain
\begin{align}
    |v_1|^2 + |v_2|^2 &\geq 3^{-\bar{f}} r^2_1.
\end{align}
Remembering that $|v_3|^2 = 3^f - (|v_1|^2 + |v_2|^2)$, this equation simplfies by noting that $3^{2 \lceil f /2\rceil -\bar{f}} = 3^f$ to yield
\begin{align}
    |v_3| &\leq 3^{f/2} \varepsilon \sqrt{1 - \varepsilon^2/4} = \mathcal{O}\left(3^{f/2} \varepsilon\right).
\end{align}

\end{document}